\documentclass[11pt]{article}
\usepackage[all]{xy}
\usepackage{amssymb}
\usepackage{amsmath}
\usepackage{graphicx}
\usepackage{epstopdf}
\usepackage{multirow}
\usepackage{amsthm}

\newtheorem{proposition}{Proposition}
\newtheorem{corollary}{Corollary}
\newtheorem{conjecture}{Conjecture}
\newtheorem{lemma}{Lemma}[]

\newtheorem{definition}{Definition}
\usepackage{authblk}
\usepackage[plain]{algorithm2e}
\date{}
\begin{document}

\title{\bf Algorithms for the uniqueness of the longest common subsequence}
\author[1,2,*]{Yue Wang}
\affil[1]{Department of Computational Medicine, University of California, Los Angeles, California, United States of America}
\affil[2]{Irving Institute for Cancer Dynamics and Department of Statistics, Columbia University, New York, New York, United States of America}
\affil[*]{E-mail address: yw4241@columbia.edu. ORCID: 0000-0001-5918-7525}

\maketitle{}

\begin{abstract}
Given several number sequences, determining the longest common subsequence is a classical problem in computer science. This problem has applications in bioinformatics, especially determining transposable genes. Nevertheless, related works only consider how to find one longest common subsequence. In this paper, we consider how to determine the uniqueness of the longest common subsequence. If there are multiple longest common subsequences, we also determine which number appears in all/some/none of the longest common subsequences. We focus on four scenarios: (1) linear sequences without duplicated numbers; (2) circular sequences without duplicated numbers; (3) linear sequences with duplicated numbers; (4) circular sequences with duplicated numbers. We develop corresponding algorithms and apply them to gene sequencing data.

\begin{flushleft}
\end{flushleft}
\end{abstract}

\section{Introduction}
Given some number sequences, a common subsequence is a number sequence which appears in all these sequences (not necessarily consecutive). Determining the longest common subsequence (LCS) for some number sequences is a classical problem in computer science. LCS is a common tool to evaluate the difference among different sequences. For example, LCS can be applied to computational linguistics \cite{lin2004automatic,sorokin2016using,silfverberg2018computational}. In biology, it is common to use the length of LCS as a quantitative score for comparing DNA sequences \cite{chen2004space,imbeault2017krab,zimin2017hybrid}. LCS has also been used to define ultraconserved elements \cite{reneker2012long} or remove incongruent markers in DNA sequences \cite{diop2020pseudomolecule}.

Various scenarios for the LCS problem have been studied. Here we list Scenarios A-E, where the first two are more commonly studied. For more works in these scenarios, readers may refer to more thorough reviews \cite{bergroth2000survey,huang2004fast,wei2020path}. 

Scenario A considers two sequences with possibly repeated numbers, and the sequence length is $n$. The goal is to find the LCS. If a number appears multiple time in a common subsequence, all appearances are counted when calculating the length of this common subsequence. This can be solved by dynamic programming with $\mathcal{O}(n^2)$ time complexity and $\mathcal{O}(n)$ space complexity \cite{hirschberg1975linear}, or more wisely with $\mathcal{O}(n^2/\log n)$ time complexity \cite{masek1980faster}, but $\mathcal{O}(n^{2-\epsilon})$ time complexity for any $\epsilon>0$ is impossible \cite{backurs2015edit}. This also can be solved with $o(n)$ space complexity and $\mathcal{O}(n^3)$ time complexity \cite{kiyomi2021longest}. 

In Scenario B, there are $m$ sequences with possibly repeated numbers, and the sequence length is $n$. The goal is to find the LCS. If a number appears multiple time in a common subsequence, all appearances are counted when calculating the length of this common subsequence. A standard dynamic programming algorithm has $\mathcal{O}(n^m)$ time complexity \cite{blum2021solving}. There have been other faster algorithms \cite{wang2010fast,mousavi2012improved,islam2019chemical}. When $m$ is not fixed, this scenario is equivalent to the maximum clique problem in graph theory, which is NP-hard \cite{maier1978complexity}, but has relatively fast exact and heuristic algorithms \cite{jiang2016combining,li2017minimization,wang2016two}. 

Scenario C considers $2$ sequences with possibly repeated numbers, and the sequence length is $n$. The goal is to find the LCS, where each number appears at most once. This scenario is NP-hard \cite{adi2010repetition}. 

Scenario D is similar to Scenario B, but only consider common subsequences that contain or do not contain certain strings \cite{wang2019efficient,ngomade2020dominant}. 

In Scenario E, the sequences are arc-annotated, and LCS should have the same arc annotation in original sequences \cite{jiang2000longest}.

In this paper, the motivation of studying the LCS problem is to apply it to compare gene sequences. Assume we have some gene sequences from different individuals of the same species or different species. Some genes are relatively unstable, and they can change their relative locations in the gene sequence (transposable). An unstable gene might also be duplicated or deleted. Therefore, these gene sequences from different individuals are not identical. Then we can find the LCS, which is useful for measuring the stability of genes. Genes in the LCS should be more stable, and genes not in the LCS should be transposable. 

Due to the motivation of comparing gene sequences, we consider four scenarios that are different from the previously studied LCS problems. These four scenarios are determined by two factors: whether the considered species has linear or circular gene sequences, and whether genes have multiple copies. When genes have multiple copies, we only consider common subsequences that consist of all or none of copies of the same gene. Scenario 1 has linear sequences without duplicated genes; Scenario 2 has circular sequences without duplicated genes; Scenario 3 has linear sequences with duplicated genes; Scenario 4 has circular sequences with duplicated genes. 

{Since LCS is commonly used as a tool for quantifying the difference among sequences, in such situations, we just need the length of the LCS. Although researchers have noticed that the LCS might not be unique \cite{apostolico1987longest,lember2014optimal}, the length of different LCS is the same. In some situations, we concern not just the length of the LCS, but also each element in the LCS. Most known methods only aim at finding one LCS. Therefore, when the LCS is not unique, these methods might produce different LCSs, which might cause confusion. For signal transmission in noise environment, the same signal can be transmitted several times, and one receives different versions (with different noises) of this signal \cite{ullsperger2006does}. We believe that words in the LCS of the received signals should be more reliable. When the LCS is not unique, LCS can be used to quantify the reliability of each word: words appearing in all LCSs should be most reliable; words appearing in some LCSs should be somewhat reliable, depending on the number of LCSs it appears in; words appearing in none of the LCSs should be least reliable. When we concern the stability of genes, and the LCS of gene sequences is not unique, the relation of gene and LCS is also important. A gene that appears in all the LCSs is highly stable; a gene that appears in some LCSs is moderately stable; a gene that appears in no LCS is unstable. Although we can find all the LCSs \cite{bergroth2000survey}, in some situations, determining all LCSs is too time-consuming, since there might be exponentially many LCSs.} 
For example, consider two sequences $(1,2,3,4,5,6,\ldots,2n-1,2n)$ and $(2,1,4,3,6,5,\ldots,2n,2n-1)$. Although the sequence length is $2n$, and the LCS length is $n$, the number of LCSs is $2^n$. To determine the relationship between genes and LCSs, we develop corresponding algorithms with polynomial time complexities for Scenarios 1, 2 (Algorithms~\ref{alg2}, \ref{alg4}). To our knowledge, there are no other determinations of whether genes appear in all LCSs with polynomial complexities. 

{For most prokaryotes, the DNA is circular, and we need to treat it as a circular sequence that can rotate, not a linear sequence.} Although circular sequences are commonly studied in the context of genomic rearrangements, they are rare in the literature of the LCS problems. {We only know some work that find the LCS for $m=2$ circular sequences \cite{maes1990cyclic}. Therefore, our Algorithm~\ref{alg3} that finds one LCS for Scenario 2 with $m\ge 3$ should also be novel. Besides nucleic acid sequences, some amino acid sequences (proteins) are circular \cite{craik2012thematic,grossi2016circular}. Also, the 1-D boundary of an object in a 2-D image is generally circular, making the comparison of circular sequences important in shape recognition \cite{bunke1993applications,lin2015circular}. Therefore, our Algorithm~\ref{alg3} has potential applications in protein comparison and pattern recognition.}

{In general, the same gene might have multiple copies in a gene sequences (Scenarios 3, 4). We want to find which gene (a specific DNA sequence) has the ability to change its position, not a certain copy of a gene that changes its position. In other words, ``transposable'' is defined for genes, not gene copies \cite{bourque2018ten}. Therefore, we should only consider common subsequences that consist of all or none copies of the same gene. When calculating the length of a common subsequence, we should count genes, not gene copies. In other words, we want to exclude the minimal number of genes, and copies of the remaining genes form the same subsequence. We develop the equivalence of Scenario 3 with the maximum clique problems on graphs (Proposition~\ref{p1}). We prove that Scenario 4 is at least as hard as the maximum clique problems on graphs, and it can be reduced to the maximum clique problems on $3$-uniform hypergraphs under a weak condition (Propositions~\ref{p2},~\ref{p3}). Even for two sequences, Scenarios 3, 4 are NP-hard, while Scenario B has polynomial algorithms. Thus they are essentially different from the classic Scenario B, which considers any subsequence. Our problem of determining the LCS in subsequences with all or none copies of the same gene has other applications, such as in linguistics. In some languages, certain words can change their relative positions in sentences (anastrophe) \cite{shahin2015anastrophe,cuzzolin2021note}. For multiple sentences expressing the same meaning, we can find the LCS, and words not in the LCS should be anastrophic. However, being anastrophic is a property of certain words, not of certain copies of certain words. Therefore, we should use the minimal number of words (not copies of words) to explain different orders of words. This means the problem of Scenario 3. Otherwise, if we consider all subsequences and count the length by copies, as in the classic Scenario B, then it is possible that one very common word that can translocate is contained in the LCS, since it appears too many times.}

If we only need to find one LCS, then Scenario 1 is a special case of Scenario B, and our method (Algorithm~\ref{alg1}) can be easily derived from standard algorithms. Scenarios 3, 4 can be reduced to maximum clique problems in graphs and hypergraphs, which are NP-hard. Although there have been numerous algorithms for the maximum clique problem \cite{wu2015review}, for the sake of completeness, we design fast heuristic algorithms (Algorithms~\ref{alg5}, \ref{alg7}) and test them to find that they only fail in rare cases. 

We proposed the idea of using the LCS to find transposable genes and Algorithm~\ref{alg1} in a previous paper \cite{kang2014flexibility}, where Algorithm~\ref{alg1} was applied to study the ``core-gene-defined genome organizational framework'' (the complement of transposable genes) in various bacteria, and it was found that for different species, the transposable gene distribution and developmental traits are correlated. This paper considers other situations (especially when the LCS is not unique), and can be regarded as a theoretical sequel of that previous paper. Algorithm~\ref{alg1} is contained in this paper for the sake of completeness.

In sum, our main contributions are Algorithms~\ref{alg2}, \ref{alg3} (for $m\ge 3$), \ref{alg4} in Scenarios 1, 2 and Propositions~\ref{p1},~\ref{p2},~\ref{p3} and Corollary~\ref{coro} in Scenarios 3, 4. We test Algorithms~\ref{alg1}--\ref{alg7} on gene sequences of different \emph{Escherichia coli} individuals and find some possible transposable genes.

We first introduce the background of transposable genes in Section~\ref{bio}. Then we
describe the setup for the LCS problem we study in Section~\ref{setup}. In Sections~\ref{s1}--\ref{s4}, we transform the LCS problem into corresponding graph theory problems and design algorithms. We finish with conclusions and discussions in Section~\ref{con}. In Appendix~\ref{app}, we apply our algorithms from Scenarios 1--4 to gene sequences of \emph{Escherichia coli} individuals and find some possible transposable genes. All the algorithms in this paper have been implemented in Python. For the code and data files, see https://github.com/YueWangMathbio/Transposon.

\section{Biological background of transposable genes}
\label{bio}
In this section, we review how gene sequences become different, and introduce the specific biological problem we want to study. We also explain how Scenarios 1--4 of the LCS problem are derived from the biological problem.

The nucleotide sequence can be changed by various events, such as inversion, insertion, deletion, and duplication \cite{ivics2010expanding}. Such rearrangement events lead to the existence of transposons (also called transposable elements or jumping genes), which are DNA sequences that can change their relative positions within the genome. Transposons were first discovered in maize by Barbara McClintock \cite{mcclintock1950origin}. Transposons have various types: long terminal repeats (LTR) retrotransposons, Dictyostelium intermediate repeat sequence (DIRS)-like elements, Penelope-like elements (PLE), long interspersed elements (LINE), short interspersed elements (SINE), terminal inverted repeats (TIR), Helitrons, etc. \cite{makalowski2019transposable}.

Transposons are common in various species. For the human genome, the proportion of transposons is approximately 44\%, although most of transposons are inactive \cite{mills2007transposable}. Transposons can participate in controlling gene expression \cite{zhou2020dna}, and they are related to several diseases, such as cancer \cite{denicola2015utility}, hemophilia \cite{kazazian1988haemophilia}, and porphyria \cite{mustajoki1999insertion}. Transposons can drive rapid phenotypic variations, which cause complicated cell behaviors \cite{zhou2014multi,niu2015phenotypic,niu2019transposable,chen2016overshoot,jiang2017phenotypic}. Transposons can be used to detect cancer drivers \cite{noorani2020crispr} and potential therapies \cite{angelini2022model}. Transposons are also essential for the development of \emph{Oxytricha trifallax} \cite{nowacki2009functional}, antibiotic resistance of bacteria \cite{babakhani2018transposons}, and the proliferation of various cells \cite{rahrmann2009identification,xia2020pde,dessalles2022naive}. With the presence of transposons, the regulation between genes might be affected, which is a challenge for inferring the structures of gene regulatory networks \cite{wang2022inference} and general transcriptome analysis \cite{sha2020inference,zhou2021dissecting}. 

When transposons have been determined, we can use them to compare the genomes of different species, and such comparisons can be combined with other measurements between species, such as metrics on developmental trees \cite{wang2022two}. Such comparisons can be also extended to different tissues to help with the prediction of tissue transplantation experiments \cite{wang2021inference}. Besides, for some species, cells at different positions have different gene expression patterns, which might be related to transposons \cite{wang2020biological}. 

Many transposons are as short as $10^2-10^3$ base pairs, shorter than a general gene \cite{payer2019transposable}. To determine such short transposons, one needs to analyze the original AGCT nucleotide sequences. There have been many algorithms developed to determine short transposons from nucleotide sequences, such as MELT (Mobile Element Locator Tool) \cite{gardner2017mobile}, ERVcaller (Endogenous RetroVirus caller) \cite{chen2019ervcaller}, and TEMP2 (Transposable Elements Movements Present 2) \cite{yu2021benchmark}. Different algorithms may only determine certain types of transposons. For more details, readers may refer to other papers \cite{orozco2020measuring,goubert2022beginner}. They use raw DNA sequencing data, which only contain imperfect information about the true DNA sequence, and the data quality depends on some factors that vary across different datasets \cite{evrony2021applications}. Besides, they need a corresponding genome or reference transposon libraries. 

There are gross DNA changes that associate with many genes, also called genomic rearrangements \cite{gu2008mechanisms}. Such rearrangements include inversion, transposition, fusion, and fission \cite{bohnenkamper2021computing}. To determine such gross genomic rearrangements, one first needs to convert nucleotide sequences into gene sequences by annotation. For two different gene sequences, the general idea of determining rearrangements is to calculate the minimal number of operations required for transforming one sequence into the other \cite{tesler2002efficient}. This defines an editing distance between gene sequences, which can be used to compare the evolution distance between species and construct the phylogenetic tree \cite{terauds2019maximum}. There have been many algorithms developed to determine genomic rearrangements. They consider different scenarios: whether the gene sequence is linear or circular, whether genes have unique labels, and what operations can be taken. Kececioglu and Sankoff only consider inversion for linear sequences with unique gene labels \cite{kececioglu1995exact}; Blanchette et al. consider inversion and transposition for circular sequences with unique gene labels \cite{blanchette1996parametric}; Tesler considers inversion, transposition, fusion, and fission for linear and circular sequences with unique gene labels \cite{tesler2002efficient}; Terauds and Sumner study circular sequences with representation theory tools \cite{terauds2019maximum}; Bohnenk{\"a}mper et al. consider linear and circular sequences with possibly duplicated labels \cite{bohnenkamper2021computing}. There are also systematic pipelines for determining rearrangements from whole-genome assemblies \cite{goel2019syri,mitsuhashi2020pipeline}. Nevertheless, these methods consider large-scale rearrangements, and minimize the number of operations to transform one gene sequence into the other, not concrete genes that can change their locations. Besides, these methods only compare two gene sequences, not more. Their results depend on the set of possible operations, which is somewhat arbitrary. 

In this paper, we consider a mesoscopic scenario between the genomic rearrangement situation and the short transposon situation: \emph{Given accurately annotated gene sequences (not nucleotide sequences) from different individuals, determine individual genes (not short nucleotide segments or long gene strands) that can change their locations (transposable).} This provides a qualitative description for the stability of genes, which can guide gene editing \cite{wang2022horizontal} and phylogenetics \cite{kang2014flexibility}. The proportion of fixed genes quantifies the robustness of the genome. We aim at minimizing the number of genes to move. When there are only two gene sequences, this is equivalent to calculating genomic arrangements, where the only allowed operation is single-gene transposition.

In the copy-paste (duplication) case and deletion case, we can compare the numbers of copies of genes for different individuals to determine the transposable genes that have changed their copy numbers. In the inversion case, we can check the direction of genes to determine transposable genes that have changed their orientations \cite{lin2011changes}. In the cut-paste (insertion) case, the compositions of gene sequences are the same, but the orders of genes differ. It is not straightforward to uniquely determine which genes have changed their relative locations. In this case, we need to introduce the LCS problem. 

\section{Problem setup}
\label{setup}
Given raw DNA sequencing data, the first step is to transform them into gene sequences. This can be done with various genome annotation tools \cite{salzberg2019next,bruuna2021braker2}. For simplicity, we replace the gene names by numbers $1,\ldots,n$.

For some species, the DNA is a line \cite{rowley2018organizational}. We can represent this DNA as a linear gene sequence of distinct numbers that represent genes: $(1,2,3,4)$. If some genes change their transcriptional orientations, we can simply detect them and handle the remaining genes. Now a linear DNA naturally has a direction (from 5' end to 3' end), thus $(1,2,3,4)$ and $(4,3,2,1)$ are two different gene sequences. 

Consider two linear gene sequences from different individuals: $(1,2,3,4)$ and $(1,4,2,3)$. We can intuitively detect that gene $4$ changes its relative position, and should be regarded as a transposable gene. However, changing the positions of genes $2,3$ can also transform one sequence into the other. The reason that we think gene $4$ (not genes $2,3$) changes its relative position is that the number of genes we need to move is smaller. Nevertheless, the number of genes that change their relative locations is difficult to determine. We can consider the complement of transposable genes, i.e., genes that do not change their relative positions. These fixed genes can be easily defined as the LCS of the given gene sequences. Here a common subsequence consists of some genes (not necessarily adjacent, different from a substring) that keep their relative orders in the original sequences. \emph{Thus transposable genes are the complement of this LCS.} Notice that the LCS might not be unique. We classify genes by their relations with the LCS(s). The motivation of classifying transposable genes with respect to the intersection and union of LCSs is similar to defining essential variables with Markov boundaries in causal inference \cite{wang2020causal}.

\begin{definition}
	A gene is \textbf{proper-transposable} if it is not contained in any LCS. A gene is \textbf{non-transposable} if it is contained in every LCS. A gene is \textbf{quasi-transposable} if it is contained in some but not all LCSs.
	\label{def1}
\end{definition}

In the example of $(1,2,3,4)$ and $(1,4,2,3)$, the unique LCS is $(1,2,3)$. Thus $4$ is proper-transposable, and $1,2,3$ are non-transposable. In the following, we consider other scenarios, where the proper/quasi/non-transposable genes still follow Definition~\ref{def1}, but the definition of the LCS differs.

For some species, the DNA is a circle, not a line \cite{verma2019architecture}. A circular DNA also has a natural direction (from 5' end to 3' end), and we use the clockwise direction to represent this natural direction. In the circular sequence scenario, a common subsequence is a circular sequence that can be obtained from each circular gene sequence by deleting some genes. See Fig.~\ref{sce2} for two circular gene sequences and their LCS. Notice that we can rotate each circular sequence for a better match.

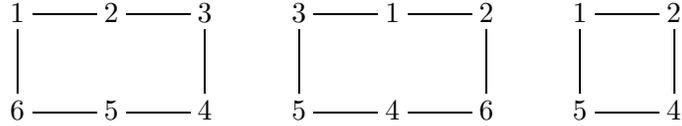
\begin{figure}[htb]
	\begin{center}
		$\xymatrix{
			1\ar@{-}[r]\ar@{-}[d]&2\ar@{-}[r]	& 3\ar@{-}[d] &3 \ar@{-}[r]\ar@{-}[d]&1\ar@{-}[r]&2\ar@{-}[d]&1\ar@{-}[r]\ar@{-}[d]&2\ar@{-}[d]\\
			6&	5\ar@{-}[l] & 4\ar@{-}[l]    &5 &4\ar@{-}[l]  &6   \ar@{-}[l]  &5&4\ar@{-}[l] 
		}$
	\end{center}
	\caption{Two circular gene sequences without duplicated genes and their LCS, corresponding to Scenario 2.}
	\label{sce2}
\end{figure}

A gene might have multiple copies (duplicated) in a gene sequence \cite{ibal2019information}. Notice that the definition of the transposable gene is a gene (specific DNA sequence) that has the ability to change its position, not a certain copy of a gene that changes its position. This means transposable genes should be defined for genes, not gene copies. Thus we should only consider common subsequences that consist of all or none copies of the same gene. When calculating the length of a common subsequence, we should count genes, not gene copies. Consider two linear sequences $(4,1,2,1,1,3,2,4,1,1)$ and $(4,1,2,3,1,1,2,1,1,4)$. If we consider any subsequences, the LCS is $(4,1,2,1,1,2,1,1)$; if we only consider subsequences that contain all or none copies of the same gene, but count the length by copies, the LCS is $(1,2,1,1,2,1,1)$; if we only consider subsequences that contain all or none copies of the same gene, and count the length by genes, the unique LCS is $(4,2,3,2,4)$, and gene $1$ is proper-transposable.

When we consider circular gene sequences with duplicated genes, we should still only consider subsequences that consist of all or none copies of the same gene, and calculate the length by genes. Notice that circular sequences can be rotated. See Fig.~\ref{sce4} for two circular gene sequences with duplicated genes and their LCS.

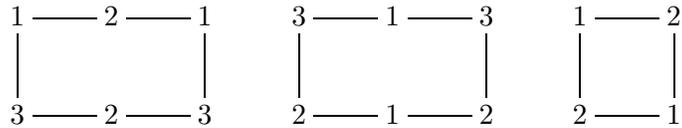
\begin{figure}[htb]
	\begin{center}
		$\xymatrix{
			1\ar@{-}[r]\ar@{-}[d]&2\ar@{-}[r]	& 1\ar@{-}[d] &3 \ar@{-}[r]\ar@{-}[d]&1\ar@{-}[r]&3\ar@{-}[d]&1\ar@{-}[r]\ar@{-}[d]&2\ar@{-}[d]\\
			3&	2\ar@{-}[l] & 3\ar@{-}[l]    &2 &1\ar@{-}[l]  &2   \ar@{-}[l]  &2&1\ar@{-}[l] 
		}$
	\end{center}
	\caption{Two circular gene sequences with duplicated genes and their LCS, corresponding to Scenario 4.}
	\label{sce4}
\end{figure}

We have turned the problem of determining transposable genes into finding the LCS of several gene sequences. Depending on whether the gene sequences are linear or circular, and whether genes have multiple copies, the problem can be classified into four scenarios:

\noindent \textbf{Scenario 1}: Consider $m$ linear sequences of genes $1,\ldots,n$, where each gene has only one copy in each sequence. Determine the longest linear sequence that is a common subsequence of these $m$ sequences.

\noindent \textbf{Scenario 2}: Consider $m$ circular sequences of genes $1,\ldots,n$, where each gene has only one copy in each sequence. Determine the longest circular sequence that is a common subsequence of these $m$ sequences. Here circular sequences can be rotated.

\noindent \textbf{Scenario 3}: Consider $m$ linear sequences of genes $1,\ldots,n$, where each gene can have multiple copies in each sequence. Determine the longest linear sequence that is a common subsequence of these $m$ sequences. Only consider subsequences that consist of all or none copies of the same gene, and calculate the length by genes.

\noindent \textbf{Scenario 4}: Consider $m$ circular sequences of genes $1,\ldots,n$, where each gene can have multiple copies in each sequence. Determine the longest circular sequence that is a common subsequence of these $m$ sequences. Only consider subsequences that consist of all or none copies of the same gene, and calculate the length by genes. Here circular sequences can be rotated.

These four scenarios correspond to different algorithms, and will be discussed separately.

\section{Linear sequences without duplicated genes}
\label{s1}
In Scenario 1, consider $m$ linear gene sequences, where each sequence contains $n$ genes $1,\ldots,n$. Each gene has only one copy. For such permutations of $1,\ldots,n$, we need to find the LCS. 

\subsection{A graph representation of the problem}
Brute-force searching that tests whether each subsequence appears in all sequences is not applicable, since the time complexity is exponential in $n$. To develop a polynomial algorithm, we first design an auxiliary directed graph $\mathcal{G}$.

\begin{definition}
	For $m$ linear sequences with $n$ non-duplicated genes, the corresponding \textbf{auxiliary graph} $\mathcal{G}$ is a directed graph, where each vertex is a gene $g_i$, and there is a directed edge from $g_i$ to $g_j$ if and only if $g_i$ appears before $g_j$ in all $m$ sequences.
\end{definition}

A directed path $g_1\to g_2\to g_3\to\cdots\to g_4\to g_5$ in $\mathcal{G}$ corresponds to a common subsequence $(g_1,g_2,g_3,\ldots,g_4,g_5)$ of $m$ sequences, and vice versa. We add $0$ to the head of each sequence and $n+1$ to the tail. Then the LCS must start at $0$ and end at $n+1$. \emph{The problem of finding the LCS becomes finding the longest path from $0$ to $n+1$ in $\mathcal{G}$.} See Fig.~\ref{ag} for an example of using the auxiliary graph to determine transposable genes.  This auxiliary graph $\mathcal{G}$ has no directed loop (acyclic). If there exists a loop $g_1\to g_2\to g_3\to\cdots\to g_4\to g_1$, then $g_1$ is prior to $g_4$ and $g_4$ is prior to $g_1$ in all sequences, a contradiction.

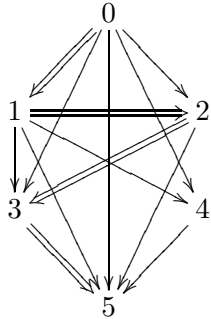
\begin{figure}[htb]
	\begin{center}
		$\xymatrix{
			& 0\ar@2{->}[dl]\ar[dr]\ar[ddl]\ar[ddr]\ar[ddd] &\\
			1\ar[d]\ar@2{->}[rr]\ar[drr]\ar[ddr] &     & 2 \ar@2{->}[dll]\ar[ddl]  \\
			3\ar@2{->}[dr]     &     & 4\ar[dl]      \\
			&5&         }$
	\end{center}
	\caption{The auxiliary graph $\mathcal{G}$ of two sequences $([0],1,2,3,4,[5])$ and $([0],1,4,2,3,[5])$. The unique longest path (double arrows) from $0$ to $5$ is $0\to 1\to 2\to 3\to 5$, meaning that the unique LCS is $([0],1,2,3,[5])$. Thus $1,2,3$ are non-transposable, and $4$ is proper-transposable.}
	\label{ag}
\end{figure}

\subsection{Find the longest path}
Determining the longest path between two vertices in a directed acyclic graph can be solved by a standard dynamic programming algorithm. For a vertex $g_i\in\{0,1,\ldots,n\}$, consider the longest path from $g_i$ to $n+1$. Since there exists an edge $g_i\to n+1$, and $\mathcal{G}$ is acyclic, this longest path exists. If the longest path is not unique, assign one arbitrarily. 
\begin{definition}
	Define $F_+(g_i)$ to be the length of the longest path from $g_i$ to $n+1$ in $\mathcal{G}$, and $H_+(g_i)$ to be the vertex next to $g_i$ in this path. 
\end{definition}
$F_+$ and $H_+$ can be calculated recursively: For one gene $g_i$, consider all genes $g_j$ with an edge $g_i\to g_j$ in $\mathcal{G}$. The gene $g_j$ with the largest $F_+(g_j)$ is assigned to be $H_+(g_i)$, and $F_+(g_i)=F_+(g_j)+1$. If $g_l\to n+1$ is the only edge that starts from gene $g_l$, then $F_+(g_l)=1$, and $H_+(g_l)=n+1$. In other words,
\[H_+(g_i)={\mathrm{argmax}}_{\{g_j\mathrm{ with }g_i\to g_j\}}\,F_+(g_j);\]
\[F_+(g_i)=1+F_+[H_+(g_i)].\]
Then $0\to H_+(0)\to H_+^{2}(0)\to H_+^{3}(0)\to 
\cdots{}\to H_+^{f-1}(0)\to H_+^{f}(0)=n+1$, denoted by $\mathcal{L}_0$, is a longest path in $\mathcal{G}$. Here $f=F_+(0)$, and $H_+^i$ is the $i$th iteration of $H_+$.

\subsection{Determine proper- and quasi-transposable genes}
\begin{definition}
	For $g_i\in\{1,\ldots,n, n+1\}$, define $F_-(g_i)$ to be the length of the longest path from $0$ to $g_i$ in $\mathcal{G}$, and $H_-(g_i)$ to be the vertex prior to $g_i$ in this path. 
\end{definition}
$F_-$ and $H_-$ can be calculated similar to $F_+$ and $H_+$. We can see that $F_+(g_i)+F_-(g_i)$ is the length of 
\[0=H_-^{F_-(g_i)}(g_i)\to H_-^{F_-(g_i)-1}(g_i) \to \cdots \to H_-(g_i) \to g_i\]
\[\to H_+(g_i) \to \cdots\to H_+^{F_+(g_i)-1}(g_i)\to H_+^{F_+(g_i)}(g_i)=n+1,\]
a longest path from $0$ through $g_i$ to $n+1$. {Therefore, a gene $g_i$ is proper-transposable if and only if $F_+(g_i)+F_-(g_i)<F_+(0)$.
	
	To further distinguish quasi-transposable genes and non-transposable genes, we need the following lemma.
	\begin{lemma}
		In Scenario 1 of linear sequences without duplicated genes, each quasi-transposable gene $g_i$ has a corresponding quasi-transposable gene $g_j$, so that $F_+(g_i)=F_+(g_j)$.
		\label{l1}
	\end{lemma}
	\begin{proof}
		In a longest path from $0$ to $n+1$, if gene $g_i$ is in the $c$th position counted from the end, then from the definition of $F_+(g_i)$, we can directly see that $F_+(g_i)\ge c-1$. If $F_+(g_i)>c-1$, we can use a path from $g_i$ to $n+1$ with length $F_+(g_i)$ to replace the path from $g_i$ to $n+1$ in the longest path, and obtain a longer path, a contradiction. Since $\mathcal{G}$ has no cycle, this path does not visit the same vertex twice. Thus we must have $F_+(g_i)=c-1$. The position of $g_i$ in a longest path from $0$ to $n+1$ is fixed.
		
		If $g_i$ is a quasi-transposon, choose a longest path from $0$ to $n+1$ that does not contain $g_i$. The gene $g_j$ at the $[F_+(g_i)+1]$th position counted from the end must have $F_+(g_j)=F_+(g_i)$, due to its position. In a longest path from $0$ to $n+1$ that contains $g_i$, $g_i$ takes the only position that $g_j$ can take, and this path cannot contain $g_j$. Therefore, $g_j$ is also quasi-transposable.
	\end{proof}
	
	After excluding proper-transposable genes, we just need to compare $F_+(\cdot)$ for remaining genes. Those with unique values of $F_+(\cdot)$ are non-transposable, and others that share the same value of $F_+(\cdot)$ with other genes are quasi-transposable.
	
	Lemma~\ref{l1} provides a restriction on the number of quasi-transposable genes: for any LCS, the number of quasi-transposable genes in this LCS is no larger than the number of quasi-transposable genes not in this LCS. The reason is that each quasi-transposable gene $g_i$ has a corresponding quasi-transposable gene $g_j$ with $F_+(g_i)=F_+(g_j)$. The possible position of $g_j$ has been occupied by $g_i$, and $g_j$ cannot be in the LCS.
}

\subsection{Algorithms and complexities}

\begin{algorithm}[!htbp]
	\caption{Detailed workflow of determining proper-transposable genes and quasi-transposable genes in Scenario 1, preparation stage. One LCS is also outputted.}
	\label{alg1}
	\vspace{-\bigskipamount}
	\ \\
	\begin{enumerate}
		{	\item \textbf{Input} 
			
			\quad $m$ linear sequences of genes $1,\ldots,n$. No duplicated genes.
			
			\item \textbf{Modify} the sequences:
			
			\quad Add $0$ to the head, and $n+1$ to the tail of each sequence 
			
			\item \textbf{Construct} the auxiliary graph $\mathcal{G}$:
			
			\quad Vertices of $\mathcal{G}$ are all the genes $1,\ldots,n$
			
			\quad \textbf{For} each pair of genes $g_i,g_j$
			
			\quad\quad   \textbf{If} $g_i$ is prior to $g_j$ in all $m$ sequences
			
			\quad\quad\quad \textbf{Add} a directed edge $g_i\to g_j$ in $\mathcal{G}$
			
			\quad\quad\textbf{End} of if
			
			\quad\textbf{End} of for
			
			\item \textbf{Calculate} $F_+(\cdot)$ and $H_+(\cdot)$ for each gene $g_i$ in $0,1,\ldots,n$ recursively; \textbf{calculate} $F_-(\cdot)$ and $H_-(\cdot)$ for each gene $g_i$ in $1,\ldots,n, n+1$ recursively: 
			
			\quad $H_+(g_i)={\mathrm{argmax}}_{\{g_j\mathrm{ with }g_i\to g_j\}}\,F_+(g_j)$ 
			
			\quad \% If $g_j$ with $g_i\to g_j$ that maximizes $F_+(g_j)$ is not unique, choose one randomly
			
			\quad $F_+(g_i)=1+F_+[H_+(g_i)]$ 
			
			\quad $H_-(g_i)={\mathrm{argmax}}_{\{g_j\mathrm{ with }g_j\to g_i\}}\,F_-(g_j)$
			
			\quad \% If argmax is not unique, choose one randomly
			
			\quad $F_-(g_i)=1+F_-[H_-(g_i)]$
			
			\item \textbf{Construct} a longest path $\mathcal{L}_0$ from $0$ to $n+1$:
			
			\quad $0\to H_+(0)\to H_+^{2}(0)\to H_+^{3}(0)\to 
			\cdots{}\to H_+^{f-1}(0)\to H_+^{f}(0)=n+1$
			
			\quad \% Here $f=F_+(0)$, and $H_+^i$ is the $i$th iteration of $H_+$
			
			\item \textbf{Output} $F_+(\cdot),H_+(\cdot),F_-(\cdot),H_-(\cdot),\mathcal{L}_0$
		}
	\end{enumerate}
\end{algorithm}

\begin{algorithm}[!htbp]
	\caption{{Detailed workflow of determining proper-transposable genes and quasi-transposable genes in Scenario 1, output stage.}}
	\label{alg2}
	\vspace{-\bigskipamount}
	\ \\
	\begin{enumerate}
		{	\item \textbf{Input} 
			
			\quad $F_+(\cdot),F_-(\cdot)$ calculated from Algorithm~\ref{alg1}		
			
			\item \textbf{For} each gene $g_i$ in $1,\ldots,n$
			
			\quad \textbf{If} $F_+(g_i)+F_-(g_i)<F_+(0)$
			
			\quad\quad \textbf{Output} $g_i$ is a proper-transposable gene
			
			\quad \textbf{Else}
			
			\quad\quad \textbf{Record} the value of $F_+(g_i)$
			
			\quad \textbf{End} of if
			
			\textbf{End} of for
			
			\item \textbf{For} each gene $g_i$ in $1,\ldots,n$ that has not been determined as proper-transposable
			
			\quad \textbf{If} $F_+(g_i)$ appears more than once in the record
			
			\quad\quad \textbf{Output} $g_i$ is a quasi-transposable gene
			
			\quad \textbf{Else}
			
			\quad\quad \textbf{Output} $g_i$ is a non-transposable gene
			
			\quad \textbf{End} of if
			
			\textbf{End} of for
			
			\item\textbf{Output}: whether each gene is proper/quasi/non-transposable
			
		}
	\end{enumerate}
\end{algorithm}

We summarize the above method as Algorithms~\ref{alg1},\ref{alg2}. If we have known that the LCS is unique, then we just need to apply Algorithm~\ref{alg1}, so that genes in $\mathcal{L}_0$ are non-transposable, and genes not in $\mathcal{L}_0$ are proper-transposable. We have reported Algorithm~\ref{alg1} previously \cite{kang2014flexibility,wang2018some}. Algorithm~\ref{alg1} is kept here to make the story complete. Assume we have $m$ sequences with length $n$. The time complexities of Steps 2-5 in Algorithm~\ref{alg1} are $\mathcal{O}(m)$, $\mathcal{O}(mn^2)$, $\mathcal{O}(n)$, $\mathcal{O}(n)$. {The time complexities of Step 2 and Step 3 in Algorithm~\ref{alg2} are $\mathcal{O}(n)$ and $\mathcal{O}(n)$.} The overall time complexity of determining transposable genes in Scenario 1 by Algorithms~\ref{alg1},\ref{alg2} is $\mathcal{O}(mn^2)$. The space complexity is trivially $\mathcal{O}(mn+n^2)$. 

{When $m=2$, Algorithm~\ref{alg1} has time complexity $\mathcal{O}(n^2)$. This is slower than an algorithm by Masek and Paterson, which has time complexity $\mathcal{O}(n^2/\log n)$ \cite{masek1980faster}. Algorithm~\ref{alg2} studies whether one gene appears in all LCSs. We do not know other algorithms that solve the same problem.}

\section{Circular sequences without duplicated genes}
\label{s2}

In Scenario 2, consider $m$ circular gene sequences, where each sequence contains $n$ genes $1,\ldots,n$. Each gene has only one copy in each sequence. For such circular permutations of $1,\ldots,n$, we need to find the LCS. Assume the length of the LCS is $n-k$.

\subsection{Find one LCS}
We first randomly choose a gene $g_i$. Cut all circular sequences at $g_i$ and expand them to be linear sequences. For example, the circular sequences in Fig.~\ref{sce2} cut at $1$ are correspondingly $(1,2,3,4,5,6)$ and $(1,2,6,4,5,3)$. Using Algorithm~\ref{alg1}, we can find $\mathcal{L}_i$ that begins with $g_i$, which is one LCS of all expanded linear sequences. In the above example, the longest common linear subsequence starting from $1$ is $(1,2,4,5)$. If $g_i$ is a non-transposable gene or a quasi-transposable gene, then $\mathcal{L}_i$ (glued back to a circle) is a longest common circular subsequence. If $g_i$ is a proper-transposable gene, then $\mathcal{L}_i$ is shorter than the longest common circular subsequence. In Fig.~\ref{sce2}, gene $1$ is non-transposable, and $(1,2,4,5)$ (glued) is the longest common circular subsequence.

We do not know if $\mathcal{L}_i$ (glued) is an LCS (whether containing $g_i$ or not) for all circular sequences. If there is a longer common subsequence, it should contain genes that are not in $\mathcal{L}_i$. Consider four variables $\mathcal{L}$, $g$, $C$, and $\mathcal{S}$, whose initial values are $\mathcal{L}_i$, $g_i$, the length of $\mathcal{L}_i$, and the complement of $\mathcal{L}_i$. These variables contain information on the longest common linear subsequence that we have found during this procedure.

Choose a gene $g_j$ in $\mathcal{S}$, and cut all circular gene sequences at $g_j$. Apply Algorithm~\ref{alg1} to find $\mathcal{L}_j$, which is the longest in common subsequences that contain $g_j$. If the length of $\mathcal{L}_j$ is larger than $C$, set $\mathcal{L}$ to be $\mathcal{L}_j$, set $g$ to be $g_j$, set $C$ to be the length of $\mathcal{L}_j$, and set $\mathcal{S}$ to be the complement of $\mathcal{L}_j$. Otherwise, keep $\mathcal{L}$, $g$, $C$, and $\mathcal{S}$ still.

Choose another gene $g_l$ in $\mathcal{S}$ which has not been chosen before, and repeat this procedure. This procedure terminates when all genes in $\mathcal{S}$ have been chosen and cut. Denote the final values of $\mathcal{L}$, $g$, $C$, and $\mathcal{S}$ by $\mathcal{L}_0$, $g_0$, $C_0$, and $\mathcal{S}_0$. Here $\mathcal{S}_0$ is the complement of $\mathcal{L}_0$.

During this procedure, if the current $g$ is a proper-transposable gene, then $\mathcal{S}$ contains a non-transposable gene or a quasi-transposable gene, which has not been chosen. Thus $\mathcal{L}$, $g$, $C$, $\mathcal{S}$ will be further updated. If the current $g$ is a non-transposable gene or a quasi-transposable gene, then $C$ has reached its maximum, and $\mathcal{L}$, $g$, $C$, $\mathcal{S}$ will not be further updated. This means $\mathcal{L}_0$ is a longest common circular subsequence, and $C_0$ is the length of the LCS, $n-k$. Also, the total number of genes being chosen and cut is $k+1$. All $k$ genes in $\mathcal{S}_0$ and $g_0$ are chosen and cut. A gene $g_t$ in $\mathcal{L}_0$ (excluding $g_0$) is non-transposable or quasi-transposable, and cannot be chosen and cut. The reason is that it cannot be chosen before $g_0$ is chosen (only proper-transposable genes can be chosen before $g_0$ is chosen), and it cannot be chosen after $g_0$ is chosen ($g_t\notin \mathcal{S}_0$).   


\subsection{Determine quasi-transposable genes}
For each gene $g_p\in\mathcal{S}_0$, apply Algorithm~\ref{alg1} to calculate $C_p$, the length of the LCS that contains $g_p$. If $C_p<C_0$, $g_p$ is a proper-transposable gene. Otherwise, $C_p=C_0$ means $g_p$ is a quasi-transposable gene. We have found all proper-transposable genes. If all genes in $\mathcal{S}_0$ are proper-transposable, then all genes in $\mathcal{L}_0$ are non-transposable, and the procedure terminates. 

If $\mathcal{S}_0$ contains quasi-transposable genes, then $\mathcal{L}_0$ also has quasi-transposable genes. To determine quasi-transposable genes in $\mathcal{L}_0$, we need the following lemma.
\begin{lemma}
	In Scenario 2, choose a quasi-transposable gene $g_p$ and cut the circular sequences at $g_p$ to obtain linear sequences. A proper-transposable gene for the circular sequences is also a proper-transposable gene for the linear sequences; a non-transposable gene for the circular sequences is also a non-transposable gene for the linear sequences. 
	\label{nl}
\end{lemma}
\begin{proof}
	Consider an LCS $\mathcal{L}_p$ for linear sequences cut at $g_p$. Since $g_p$ is a quasi-transposable gene, the length of $\mathcal{L}_p$ is also $n-k$, meaning that $\mathcal{L}_p$ is also an LCS for circular sequences. Now, this lemma is proved by the definition of proper/quasi/non-transposable gene.
\end{proof}

If a gene $g_r$ in $\mathcal{L}_0$ is non-transposable for the circular sequences, then $g_r$ is a non-transposable gene for linear sequences cut at each quasi-transposable gene $g_q\in\mathcal{S}_0$. If a gene $g_s$ in $\mathcal{L}_0$ is quasi-transposable for the circular sequences, then there is a longest common circular subsequence $\mathcal{L}_t$ that does not contain $g_s$, meaning that $\mathcal{L}_t$ contains a quasi-transposable gene $g_t$ not in $\mathcal{L}_0$. Then $g_s$ is a proper/quasi-transposable gene for linear sequences cut at $g_t$. 

Therefore, we can use the following method to determine quasi-transposable genes in $\mathcal{L}_0$. For each quasi-transposable gene $g_q\in\mathcal{S}_0$, cut at $g_q$ and apply Algorithms~\ref{alg1},\ref{alg2} to determine if each gene in $\mathcal{L}_0$ is proper/quasi/non-transposable for the linear gene sequences cut at $g_q$. A gene $g_r\in\mathcal{L}_0$ is non-transposable for the circular sequences if and only if it is non-transposable for linear sequences cut at any quasi-transposable gene $g_q\in\mathcal{S}_0$. A gene $g_s\in\mathcal{L}_0$ is quasi-transposable for the circular sequences if and only if it is proper/quasi-transposable for linear sequences cut at some quasi-transposable gene $g_q\in\mathcal{S}_0$.

When we have determined all quasi-transposable genes in $\mathcal{S}_0$, it might be tempting to apply a simpler approach to determine quasi-transposable genes in $\mathcal{L}_0$: For each quasi-transposable gene $g_q\in\mathcal{S}_0$, cut at $g_q$ and apply Algorithm~\ref{alg1} to find an LCS $\mathcal{L}_q$. A gene in $\mathcal{L}_0$ is non-transposable if and only if it appears in all such $\mathcal{L}_q$. This approach is valid only if the following conjecture holds:

\begin{conjecture}
	In Scenario 2 of circular sequences without duplicated genes, each quasi-transposable gene $g_i$ has a corresponding quasi-transposable gene $g_j$, so that no LCS can contain both $g_i$ and $g_j$.
	\label{c1}
\end{conjecture}

However, Conjecture~\ref{c1} does not hold. See Fig.~\ref{ce} for a counterexample. All genes are quasi-transposable. Any two quasi-transposable genes are contained in an LCS (length $3$). Thus the simplified approach above does not work.

\begin{figure}[htb]
	\begin{center}
		$\xymatrix{
			1\ar@{-}[r]\ar@{-}[d]&2\ar@{-}[r]	& 3\ar@{-}[d] &1 \ar@{-}[r]\ar@{-}[d]&2\ar@{-}[r]&6\ar@{-}[d]&1\ar@{-}[r]\ar@{-}[d]&2\ar@{-}[r]&7\ar@{-}[d]\\
			8\ar@{-}[d]&	 & 4\ar@{-}[d]    &3\ar@{-}[d] & &5   \ar@{-}[d]  &6\ar@{-}[d]&&8\ar@{-}[d]\\
			7\ar@{-}[r]&6\ar@{-}[r]&5&4\ar@{-}[r]&7\ar@{-}[r]&8&5\ar@{-}[r]&4\ar@{-}[r]&3
		}$
	\end{center}
	\caption{A counterexample with three circular sequences that fails Conjecture~\ref{c1}.}
	\label{ce}
\end{figure}
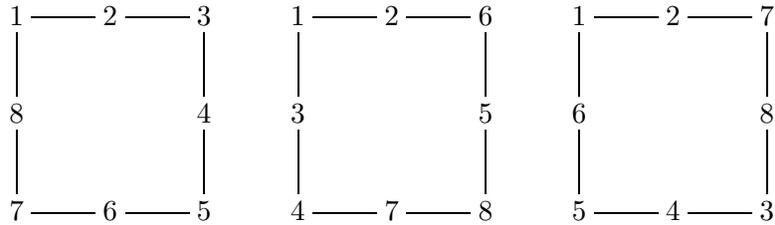

We summarize the above method as Algorithms~\ref{alg3},\ref{alg4}. If we have known that the LCS is unique, then we just need to apply Algorithm~\ref{alg3}, so that genes in $\mathcal{S}_0$ are proper-transposable, and genes not in $\mathcal{S}_0$ are non-transposable. Assume we have $m$ sequences with length $n$, and the length of the LCS is $n-k$. The time complexities of Step 2 and Step 3 in Algorithm~\ref{alg3} are $\mathcal{O}(mn^2)$ and $\mathcal{O}(kmn^2)$. The time complexities of Step 2 in Algorithm~\ref{alg4} is $\mathcal{O}(kmn^2)$. The overall time complexity of determining transposable genes in Scenario 2 by Algorithms~\ref{alg3},\ref{alg4} is $\mathcal{O}(kmn^2)$. The space complexity is trivially $\mathcal{O}(mn+n^2)$. 

{When $m=2$, Algorithm~\ref{alg3} has time complexity $\mathcal{O}(kn^2)$. If $k>\log n$, this is slower than an algorithm by Maes, which has time complexity $\mathcal{O}(n^2 \log n)$ \cite{maes1990cyclic}. Nevertheless, we do not know any reported algorithms that apply for general $m$. Algorithm~\ref{alg4} studies whether one gene appears in all LCSs. We do not know other algorithms that solve the same problem.}

\begin{algorithm}[!htbp]
	\caption{Detailed workflow of determining proper-transposable genes and quasi-transposable genes in Scenario 2, preparation stage.}
	\label{alg3}
	\vspace{-\bigskipamount}
	\ \\
	\begin{enumerate}
		{	\item \textbf{Input}
			
			\quad $m$ circular sequences of genes $1,\ldots,n$, where each gene has only one copy in each sequence
			
			\item \textbf{Choose} a gene $g_i$ randomly 
			
			\textbf{Cut} all circular sequences at $g_i$ and expand them to be linear sequences
			
			\textbf{Apply} Algorithm~\ref{alg1} to find $\mathcal{L}_i$, an LCS in the expanded linear sequences
			
			\textbf{Set} $C$ to be the length of $\mathcal{L}_i$, and \textbf{set} $\mathcal{S}$ to be the complement of $\mathcal{L}_i$
			
			\item \textbf{While} $\mathcal{S}$ has a gene $g_j$ that has not been chosen and cut
			
			\quad \textbf{Cut} all circular sequences at $g_j$ and apply Algorithm~\ref{alg1} to find $\mathcal{L}_j$
			
			\quad \textbf{Denote} the length of $\mathcal{L}_j$ by $C_j$
			
			\quad \textbf{If} $C_j>C$
			
			\quad\quad   \textbf{Update} $C$ to be $C_j$, and \textbf{update} $\mathcal{S}$ to be the complement of $\mathcal{L}_j$
			
			\quad\textbf{End} of if
			
			\textbf{End} of while
			
			\textbf{Denote} the final $C$ by $C_0$, and \textbf{denote} the final $\mathcal{S}$ by $\mathcal{S}_0$
			
			\item \textbf{Output} $C_0$ and $\mathcal{S}_0$
		}
	\end{enumerate}
\end{algorithm}

\begin{algorithm}[!htbp]
	\caption{Detailed workflow of determining proper-transposable genes and quasi-transposable genes in Scenario 2, output stage.}
	\label{alg4}
	\vspace{-\bigskipamount}
	\ \\
	\begin{enumerate}
		{	\item \textbf{Input}
			
			\quad $m$ circular sequences of genes $1,\ldots,n$, where each gene has only one copy in each sequence; $C_0$ and $\mathcal{S}_0$ calculated from Algorithm~\ref{alg3}

			\item \textbf{For} each gene $g_l\in \mathcal{S}_0$ 	
			
			\quad \textbf{Cut} all circular sequences at $g_l$ and expand them to be linear sequences
			
			\quad\textbf{Apply} Algorithm~\ref{alg1} to find $\mathcal{L}_l$, an LCS in the expanded linear sequences
			
			\quad \textbf{Denote} the length of $\mathcal{L}_l$ by $C_l$	
			
			\quad \textbf{If} $C_l<C_0$
			
			\quad\quad \textbf{Output} $g_l$ is a proper-transposable gene
			
			\quad \textbf{Else}
			
			\quad\quad \textbf{Output} $g_l$ is a quasi-transposable gene
			
			\quad\quad \textbf{Cut} all circular sequences at $g_l$ and \textbf{apply} Algorithms~\ref{alg1},\ref{alg2} to find all proper/quasi-transposable genes for linear gene sequences starting at $g_l$
			
			\quad\quad \textbf{Output} genes not in $\mathcal{S}_0$ but being proper/quasi-transposable for such linear sequences are quasi-transposable for circular sequences
			
			\quad \textbf{End} of if
			
			\textbf{End} of for		
			
			\textbf{Output} other genes that have not been determined to be proper/quasi-transposable are all non-transposable
			
			\item \textbf{Output}: whether each gene is proper/quasi/non-transposable
		}
	\end{enumerate}
\end{algorithm}

\section{Linear sequences with duplicated genes}
\label{s3}
In Scenario 3, consider $m$ linear gene sequences, where each sequence contains different numbers of copies of $n$ genes $1,\ldots,n$. We need to find the LCS. Here we only consider common subsequences that consist of all or none copies of the same gene, and the subsequence length is calculated by genes, not gene copies. {If a gene has different copy numbers in different linear sequences, especially if a gene does not appear in certain sequences, then this gene cannot be in the LCS, and we directly discard this gene from all linear sequences. In the following, assume each gene has the same copy number in all linear sequences.}

\subsection{A graph representation of the problem}
Similar to Scenario 1, we construct an auxiliary graph $\mathcal{G}$, where each vertex is a gene (not a copy of a gene). However, in this case, the auxiliary graph is undirected: There is an undirected edge between gene $g_i$ and gene $g_j$ if and only if all the copies of $g_i$ and $g_j$ keep their relative locations in all sequences. For example, consider two sequences $(1,2,3,2,3,4,5)$ and $(2,1,3,3,2,4,5)$. For gene pair $1,3$, the corresponding sequences are $(1,3,3)$ and $(1,3,3)$, meaning that there is an edge between $1$ and $3$. For gene pair $1,2$, the corresponding sequences are $(1,2,2)$ and $(2,1,2)$, meaning that there is no edge between $1$ and $2$. See Fig.~\ref{sce3} for the auxiliary graph in this case. 

\begin{figure}[htb]
	\begin{center}
		$\xymatrix{
			& 1\ar@{-}[dr]\ar@{-}[ddr]\ar@{-}[ddl] &\\
			2\ar@{-}[d]\ar@{-}[drr] &     & 3\ar@{-}[dll]\ar@{-}[d]   \\
			4\ar@{-}[rr]     &     & 5       }$
	\end{center}
	\caption{The auxiliary graph $\mathcal{G}$ of two sequences $(1,2,3,2,3,4,5)$ and $(2,1,3,3,2,4,5)$. The unique maximum clique is $\{1,3,4,5\}$, meaning that the unique LCS is $(1,3,3,4,5)$. Thus $1,3,4,5$ are non-transposable genes, and $2$ is a proper-transposable gene.}
	\label{sce3}
\end{figure}
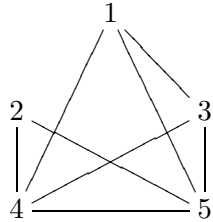

\begin{definition}
	A subgraph of $\mathcal{G}$ consists of some genes $g_1,\ldots,g_l$ and the edges between them. In a subgraph, if there is an edge between any two genes, this subgraph is called a clique. The maximum clique is the clique with the most vertices.
	\label{def5}
\end{definition}

\begin{definition}
	In graph $\mathcal{G}$, the degree of a gene $g$ is the number of edges linking $g$. In a clique of $p$ genes, where any two genes have an edge in between, each gene has degree $p-1$. 
\end{definition}

\begin{definition}
	If all copies of genes $g_1,\ldots,g_l$ keep their relative locations in all linear sequences, we say that $g_1,\ldots,g_l$ form a common subsequence. {Specifically, $g_i,g_j$ form a common subsequence if and only if $\mathcal{G}$ has an edge between $g_i,g_j$.}
	\label{def7}
\end{definition}

{The following Lemma~\ref{lenew} shows that for determining whether $g_1,\ldots,g_l$ form a common subsequence, we just need to verify each pair $g_i,g_j$ form a common subsequence.}

\begin{lemma}
	{ Genes $g_1,\ldots,g_l$ form a common subsequence if and only if each pair $g_i,g_j$ in $g_1,\ldots,g_l$ form a common subsequence.}
	\label{lenew}
\end{lemma}
\begin{proof}
	{If copies of $g_1,\ldots,g_l$ keep their relative locations in all linear sequences, then obviously copies of $g_i,g_j$ also keep their relative locations.
		
		For the other direction, only consider copies of $g_1,\ldots,g_l$ in these linear sequences. If $g_1,\ldots,g_l$ do not form a common subsequence, find the first digit that such linear sequences differ. Assume $g_p$ and $g_q$ can both appear in this digit. Then the copy sequences of $g_p,g_q$ are the same before this digit, and differ at this digit. This means $g_p,g_q$ cannot form a common subsequence. We illustrate this with Fig.~\ref{sce3}: For genes $2,3,4$, the copy sequences are $(2,3,2,3,4)$ and $(2,3,3,2,4)$. The third digit is different, where $2$ and $3$ can both appear. Then the copy sequences for genes $2,3$, $(2,3,2,3)$ and $(2,3,3,2)$, cannot match, and there is no edge between $2$ and $3$.}
\end{proof}

The following Lemma~\ref{l2} shows that there is a bijection between common subsequences of the linear sequences and cliques in $\mathcal{G}$. \emph{The problem of determining the LCS can be solved by determining the maximum clique of $\mathcal{G}$.} 

\begin{lemma}
	{Given linear sequences in Scenario 3, some genes $g_1,\ldots,g_l$ form a clique in the auxiliary graph $\mathcal{G}$ if and only if $g_1,\ldots,g_l$ form a common subsequence for the linear sequences. Therefore, the LCS for the linear sequences corresponds to the maximum clique for $\mathcal{G}$.}
	\label{l2}
\end{lemma}
\begin{proof}
	{If $l=1$, this single gene $g_1$ is a clique in $\mathcal{G}$, and it forms a common subsequence.
		
		If $l\ge 2$, the following four statements are equivalent:
		
		\noindent (1) Genes $g_1,\ldots,g_l$ form a clique in the auxiliary graph $\mathcal{G}$.
		
		\noindent (2) Each pair $g_i,g_j$ in $g_1,\ldots,g_l$ has an edge in $\mathcal{G}$.
		
		\noindent (3) Each pair $g_i,g_j$ in $g_1,\ldots,g_l$ form a common subsequence.
		
		\noindent (4) Genes $g_1,\ldots,g_l$ form a common subsequence.
		
		(1)$\Leftrightarrow$ (2) is from Definition~\ref{def5}; (2)$\Leftrightarrow$ (3) is from Definition~\ref{def7}; (3)$\Leftrightarrow$ (4) is from Lemma~\ref{lenew}.}
\end{proof}

\subsection{Equivalence to the maximum clique problem}
The above discussion shows that given gene sequences, we can construct an undirected graph $\mathcal{G}$, so that there is a bijection between common subsequences and cliques. The inverse also holds: We can construct corresponding gene sequences for a graph, {and there is a bijection between common subsequences and cliques.}

{For an undirected graph $\mathcal{G}$ with genes $1,2,\ldots,n$ as vertices, construct two auxiliary sequences $\mathfrak{S}_1$ and $\mathfrak{S}_2$ as the following: Initially, $\mathfrak{S}_1$ and $\mathfrak{S}_2$ are both $(1,2,\ldots,n)$. By a given order, check each pair $i,j$ in $1,2,\ldots,n$ and perform the following operation: if $\mathcal{G}$ has an edge between $i,j$, add $i,j$ to the end of $\mathfrak{S}_1$ and $\mathfrak{S}_2$; if $\mathcal{G}$ does not have an edge between $i,j$, add $i,j$ to the end of $\mathfrak{S}_1$, and add $j,i$ to the end of  $\mathfrak{S}_2$.}

\begin{lemma}
	{Given an undirected graph $\mathcal{G}$ and the corresponding auxiliary sequences $\mathfrak{S}_1,\mathfrak{S}_2$, some genes $g_1,\ldots,g_l$ form a clique in $\mathcal{G}$ if and only if $g_1,\ldots,g_l$ form a common subsequence for $\mathfrak{S}_1$ and $\mathfrak{S}_2$. Therefore, the LCS for the linear sequences corresponds to the maximum clique for $\mathcal{G}$.}
	\label{le}
\end{lemma}
\begin{proof}
	{ If $l=1$, $g_1$ itself is a clique in $\mathcal{G}$, and it forms a common subsequence for $\mathfrak{S}_1$ and $\mathfrak{S}_2$. 
		
		If $l\ge 2$, similar to the proof of Lemma~\ref{l2}, by Definition~\ref{def5} and Lemma~\ref{lenew}, we just need to prove that for any pair $g_i,g_j$ in $g_1,\ldots,g_l$, there is an edge between $g_i,g_j$ in $\mathcal{G}$ if and only if $g_i$ and $g_j$ form a common subsequence in $\mathfrak{S}_1$ and $\mathfrak{S}_2$.
		
		For $g_i<g_j$, consider the copy sequences $\mathfrak{C}_1$ and $\mathfrak{C}_2$ of $g_i,g_j$ in $\mathfrak{S}_1$ and $\mathfrak{S}_2$. Initially, $\mathfrak{C}_1$ and $\mathfrak{C}_2$ are both $(g_i,g_j)$. Next, when we check each gene pair to construct $\mathfrak{S}_1$ and $\mathfrak{S}_2$, there are four possibilities: (I) If the checked pair $g_p,g_q$ has neither $g_i$ nor $g_j$, then $\mathfrak{C}_1$ and $\mathfrak{C}_2$ are not affected. (II) If the checked pair is $g_i,g_k$, whether there is an edge $g_i,g_k$ in $\mathcal{G}$, $\mathfrak{C}_1$ and $\mathfrak{C}_2$ will both be added a $g_i$. (III) If the checked pair is $g_j,g_k$, whether there is an edge $g_j,g_k$ in $\mathcal{G}$, $\mathfrak{C}_1$ and $\mathfrak{C}_2$ will both be added a $g_j$. (IV) If the checked pair is $g_i,g_j$, we always add $g_i,g_j$ to $\mathfrak{C}_1$. Depending on whether edge $g_i,g_j$ exists in $\mathcal{G}$, we add $g_i,g_j$ or $g_j,g_i$ to $\mathfrak{C}_2$. In sum, if edge $g_i,g_j$ exists in $\mathcal{G}$, $\mathfrak{C}_1$ and $\mathfrak{C}_2$ are equal; if edge $g_i,g_j$ does not exist in $\mathcal{G}$, $\mathfrak{C}_1$ and $\mathfrak{C}_2$ differ by two digits.}
	
	{For example, corresponding to Fig.~\ref{sce3}, we start with $(1,2,3,4,5)$ and $(1,2,3,4,5)$. Since there is no edge between $1,2$ or between $2,3$, we have
		\[\mathfrak{S}_1=(1,2,3,4,5,\underline{1},\underline{2},1,3,1,4,1,5,\underline{2},\underline{3},2,4,2,5,3,4,3,5,4,5),\]
		\[\mathfrak{S}_2=(1,2,3,4,5,\underline{2},\underline{1},1,3,1,4,1,5,\underline{3},\underline{2},2,4,2,5,3,4,3,5,4,5).\]
		There is no edge between $2,3$ in Fig.~\ref{sce3}, and the corresponding copy sequences are $\mathfrak{C}_1=(2,3,2,3,\underline{2},\underline{3},2,2,3,3)$ and $\mathfrak{C}_2(2,3,2,3,\underline{3},\underline{2},2,2,3,3)$, which are different. There is an edge between $1,4$ in Fig.~\ref{sce3}, and the corresponding copy sequences are both $\mathfrak{C}_1=\mathfrak{C}_2=(1,4,1,1,1,4,1,4,4,4)$.}
\end{proof}

Combining Lemma~\ref{l2} and Lemma~\ref{le}, we obtain the following result:
\begin{proposition}
	{In computational complexity theory, finding the LCS in Scenario 3 is equivalent to finding the maximum clique in an undirected graph, meaning that if we have an oracle machine that solves one problem immediately, we can solve the other problem in an extra polynomial time.}
	\label{p1}
\end{proposition}
\begin{proof}
	{For an undirected graph $\mathcal{G}$ with $n$ vertices, we can construct the two auxiliary sequences $\mathfrak{S}_1$ and $\mathfrak{S}_2$ in $O(n^2)$ time. Given the LCS for $\mathfrak{S}_1$ and $\mathfrak{S}_2$ from the oracle machine, due to Lemma~\ref{le}, the genes in this LCS form the maximum clique for $\mathcal{G}$. 
		
		For $m$ linear sequences in Scenario 3 with $O(n)$ gene copies, we can construct corresponding auxiliary graph $\mathcal{G}$ in $O(mn^2)$ time. Given the maximum clique for $\mathcal{G}$ from the oracle machine, due to Lemma~\ref{l2}, the genes in this maximum clique are all the genes in the LCS for the linear sequences, and we need $O(n)$ time to construct the LCS.
		
		Therefore, finding the LCS in Scenario 3 and finding the maximum clique for an undirected graph are equivalent. }
\end{proof}

{The problem of determining the maximum clique in an undirected graph (maximum clique problem) is NP-hard \cite{valiente2002algorithms}. Therefore, from Proposition~\ref{p1}, we obtain the following result:}

\begin{corollary}
	{Finding the LCS in Scenario 3 is NP-hard.}
	\label{coro}
\end{corollary}
\begin{proof}
	{Since the maximum clique problem is NP-hard, every NP problem can be reduced to the maximum clique problem in polynomial time. From the first half of Proposition~\ref{p1}, the maximum clique problem can be reduced to the LCS problem in Scenario 3 in polynomial time. Therefore, every NP problem can be reduced to the LCS problem in Scenario 3 in polynomial time, and the LCS problem in Scenario 3 is NP-hard.}
\end{proof}
This means it is not likely to design an algorithm that always correctly determines the LCS in polynomial time. 

\subsection{A heuristic algorithm}

We have transformed Scenario 3 into the maximum clique problem for a graph $\mathcal{G}$. There have been various algorithms for the maximum clique problem \cite{jiang2016combining,li2017minimization,wang2016two}, and readers may refer to a review for more details \cite{wu2015review}. For completeness, we propose a simple idea: In the auxiliary graph $\mathcal{G}$, repeatedly abandon the gene with the smallest degree (and also edges linking this gene) until the remaining genes form a clique. See Algorithm~\ref{alg5} for the details of this greedy heuristic method. This algorithm is easy to understand, and can provide some intuition. We do not claim that Algorithm~\ref{alg5} is comparable to other sophisticated algorithms.

\begin{algorithm}[!htbp]
	\caption{A heuristic method for detecting transposable genes in Scenario 3.}
	\label{alg5}
	\vspace{-\bigskipamount}
	\ \\
	\begin{enumerate}
		{	\item \textbf{Input}
			
			\quad $m$ linear sequences of genes $1,\ldots,n$, where each gene can have multiple copies
			
			\item \textbf{Construct} the auxiliary graph $\mathcal{G}$:
			
			\quad Vertices of $\mathcal{G}$ are all the genes $1,\ldots,n$ (not their copies)
			
			\quad \textbf{For} each pair of genes $g_i,g_j$
			
			\quad\quad   \textbf{If} all copies of $g_i$ and $g_j$ keep their relative locations in all $m$ sequences
			
			\quad\quad\quad \textbf{Add} an undirected edge between $g_i$ and $g_j$ in $\mathcal{G}$
			
			\quad\quad\textbf{End} of if
			
			\quad\textbf{End} of for
			
			\textbf{Calculate} the degree for each gene in $\mathcal{G}$
			
			\item \textbf{While} true
			
			\quad\textbf{Find} a gene $g_i$ with the smallest degree $d_i$ in $\mathcal{G}$ 
			
			\quad \% If the minimal $g_i$ is not unique, choose one randomly
			
			\quad \textbf{If} $d_i+1$ is smaller than the number of genes in $\mathcal{G}$
			
			\quad \quad \textbf{Delete} $g_i$ and edges linking $g_i$ in $\mathcal{G}$ 
			
			\quad \quad \textbf{Update} the degrees of other genes
			
			\quad \textbf{Else}
			
			\quad \% The remaining genes form a clique
			
			\quad \quad \textbf{Break} the while loop
			
			\quad\textbf{End} of if
			
			\textbf{End} of while
			
			\% The final $\mathcal{G}$ is a clique of the original $\mathcal{G}$, and it is likely to be the largest one
			
			\item \textbf{Output} genes in the final $\mathcal{G}$ are not transposable, and genes not in the final $\mathcal{G}$ are transposable
		}
	\end{enumerate}
\end{algorithm}

Algorithm~\ref{alg5} does not always produce the correct result. See Fig.~\ref{ce2} for a counterexample. Here genes $1,2,3,4,5,6$ have degree $4$, while genes $7,8,9,10$ have degree $3$. When applying Algorithm~\ref{alg5}, genes $7,8,9,10$ are first abandoned, and the final result just has three genes, such as $1,3,5$. However, the maximum clique is $7,8,9,10$. Besides, Algorithm~\ref{alg5} can only determine one (possibly longest) common subsequence. Thus we cannot determine the existence of quasi-transposable genes.

\begin{figure}[htb]
	\begin{center}
		$\xymatrix{
			3\ar@{-}[d]\ar@{-}[dr]\ar@{-}[drr]\ar@{-}[r]&1\ar@{-}[dl]\ar@{-}[r]\ar@{-}[dr]&4\ar@{-}[dll]\ar@{-}[dl]\ar@{-}[d]&7\ar@{-}[d]\ar@{-}[r]\ar@{-}[dr]&8\ar@{-}[d]\ar@{-}[dl]\\
			5\ar@{-}[r]&2\ar@{-}[r]&6&10\ar@{-}[r]&9
		}$
	\end{center}
	\caption{The auxiliary graph $\mathcal{G}$ of linear sequences $(7,8,9,10,1,1,2,3,3,4,5,5,6)$ and $(1,2,1,3,4,3,5,6,5,7,8,9,10)$. This counterexample fails Algorithm~\ref{alg5}.}
	\label{ce2}
\end{figure}
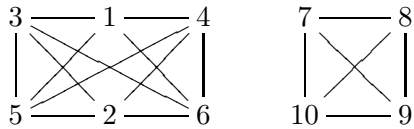

{We test Algorithm~\ref{alg5} on random graphs. Construct a random graph with $n$ genes, and any two genes have probability $p$ to have an edge between them. Use brute-force search to find the maximum clique, and compare its size with the result of Algorithm~\ref{alg5}. For each $5\le n\le 12$ and $p=0.01,0.1,0.3,0.5,0.7,0.9,0.99$, we repeat this for $10000$ times. Since finding the true maximum clique requires exponentially slow brute-force search, we do not test on very large graphs. See Table~\ref{random3} for the success rate of Algorithm~\ref{alg5}. We can see that when $p\sim 0.5$, the success rate decreases rapidly with the graph size $n$. In practice, different gene sequences do not differ too much, and most gene pairs have an edge that links them. This corresponds to $p\sim 1$ case, and Algorithm~\ref{alg5} can generally produce the correct result. Still, the performance of Algorithm~\ref{alg5} is not as good as other more sophisticated yet more complicated methods \cite{jiang2016combining,li2017minimization,wang2016two}.}

\begin{table}[]
	\begin{tabular}{llllllll}
		\multirow{2}{*}{$n$} &    &    &    & \multicolumn{3}{l}{Success rate} &  \\
		& $p=0.01$ & $p=0.1$ & $p=0.3$ & $p=0.5$ & $p=0.7$ & $p=0.9$ & $p=0.99$ \\
		5                  & 1.0000   & 1.0000   & 0.9954   & 0.9807   & 0.9821   &  0.9991  & 1.0000   \\
		6                  & 1.0000   & 0.9998   & 0.9815   & 0.9535   & 0.9667   &  0.9959  & 1.0000   \\
		7                  & 1.0000   & 0.9997   & 0.9573   & 0.9349   & 0.9464   &  0.9906  & 1.0000   \\
		8                  & 1.0000   & 0.9992   & 0.9269   & 0.9169   & 0.9269   &  0.9849  & 1.0000  \\
		9                  & 1.0000   & 0.9978   & 0.8988   & 0.8879   & 0.9062   &  0.9773  & 1.0000   \\
		10                 &  1.0000  & 0.9946   & 0.8763   & 0.8654   & 0.8692   &  0.9621  &  1.0000  \\
		11                 &  1.0000  & 0.9931   & 0.8673   & 0.8363   &  0.8468  &  0.9502  &   1.0000 \\
		12                 &  1.0000  & 0.9871   & 0.8616   & 0.8129   & 0.8215   &  0.9399  &  1.0000 
	\end{tabular}
	\caption{{The test results of Algorithm~\ref{alg5} on random graphs. We generate a graph with $n$ vertices, and for each vertex pair, there is an edge linking them with probability $p$. For each random graph, we check whether Algorithm~\ref{alg5} produces the true maximum clique. Each situation is repeated 10000 times, and we present the success rate.}}
	\label{random3}
\end{table}

{We apply Algorithm~\ref{alg5} to real gene sequencing data and manually confirm that the results are correct. See Appendix~\ref{app3}.}


Assume we have $m$ sequences with $n$ genes. In general, the copy number of a gene is small, and we can assume the length of each sequence is $\mathcal{O}(n)$. The time complexities of Step 2 and Step 3 in Algorithm~\ref{alg5} are $\mathcal{O}(mn^2)$ and $\mathcal{O}(n^2)$, and the overall time complexity is $\mathcal{O}(mn^2)$. The space complexity is trivially $\mathcal{O}(mn+n^2)$. 

{For the original LCS problem, since we only consider subsequences that contain all or none copies of the same gene, we do not know other algorithms that directly solve the same problem.
	
	We have transformed the LCS problem into the maximum clique problem for a graph. We recommend applying other algorithms for this maximum clique problem \cite{jiang2016combining,li2017minimization,wang2016two}, not Algorithm~\ref{alg5}, due to its low success rate for larger graphs. Our main contribution for this scenario is Proposition~\ref{p1} and Corollary~\ref{coro}.}

		%
		%
		%
		%
		%
		%
		%
		%
		%
		%
		%
		%
		%

\section{Circular sequences with duplicated genes}
\label{s4}
In Scenario 4, consider $m$ circular gene sequences, where each sequence contains different numbers of copies of $n$ genes $1,\ldots,n$. We need to find the LCS. Here we only consider common subsequences that consist of all or none copies of the same gene, and the subsequence length is calculated by genes, not gene copies.

We shall prove that finding the LCS in Scenario 4 is no easier than in Scenario 3. Thus Scenario 4 is also NP-hard. {If we consider all subsequences, not just subsequences that consist of all or none copies of the same gene, then it is known that the LCS problem is NP-hard \cite{nicolas2007longest}.}

\begin{proposition}
	Finding the LCS in Scenario 4 is NP-hard.
	\label{p2}
\end{proposition}
\begin{proof}
	From Corollary~\ref{coro}, Scenario 3 is NP-hard, meaning that any NP problem can be reduced to Scenario 3 in polynomial time. We just need to prove that Scenario 3 can be reduced to Scenario 4 in polynomial time: { Given a problem in Scenario 3, we can construct a problem in Scenario 4 in polynomial time. With the solution to the Scenario 4 problem, we can find the solution to the Scenario 3 problem in an extra polynomial time.}
	
	Given $m$ linear sequences with $n$ genes in Scenario 3, add genes $n+1,\ldots,2n+1$ to the end of each sequence, and glue each linear sequence into a circular sequence. The LCS for these circular sequences has the following properties {(proved later)}: (1) the circular LCS contains all genes $n+1,\ldots,2n+1$; (2) after cutting at $n+1$ and removing genes $n+1,\ldots,2n+1$ from the circular LCS, the remaining linear sequence is the LCS in Scenario 3. Therefore, if we can find the LCS for these circular sequences, then we can find the LCS for linear sequences in polynomial time.
	
	{To prove (1), notice that $(n+1,\ldots,2n+1$) is a common subsequence for the circular sequences. Therefore, the LCS has at least $n+1$ genes. This means at least one gene in $n+1,\ldots,2n+1$ is included in the LCS. Without loss of generality, assume $n+1$ is in the LCS.} Since gene $n+1$ has only one copy and it is aligned in all sequences, $n+2,\ldots,2n+1$ are also aligned, meaning that they are also in the LCS. 
	
	{To prove (2), notice that after cutting at $n+1$} and removing $n+1,\ldots,2n+1$ from the circular LCS, the remaining linear sequence is a common subsequence in Scenario 3. If there is a longer common subsequence, {then this longer common subsequence for Scenario 3 plus} $n+1,\ldots,2n+1$ should be a longer common subsequence in Scenario 4, a contradiction.

\end{proof}

Similar to Scenario 3, to find the LCS in Scenario 4, we want to reduce it to a maximum clique problem. However, Lemma~\ref{l2} does not hold in Scenario 4. For example, we can consider a circular sequence $(1,2,3)$ and its mirror symmetry $(1,3,2)$. These two sequences are different, but any two genes form a common subsequence. {Next, we modify the gene pairs in Lemma~\ref{l2} to gene triples.}

{\begin{conjecture}
		In Scenario 4, if any three genes $g_i,g_j,g_l$ in $g_1,\ldots,g_k$ form a common subsequence, then $g_1,\ldots,g_k$ form a common subsequence. 
		\label{c2}
	\end{conjecture}
	Nevertheless, Conjecture~\ref{c2} is not always true. See Fig.~\ref{cec2} for a counterexample. 
	\begin{figure}[htb]
		\begin{center}
			$\xymatrix{
				1\ar@{-}[r]&2\ar@{-}[r]&3\ar@{-}[r]&4\ar@{-}[r]&1\ar@{-}[d]&1\ar@{-}[r]&2\ar@{-}[r]&4\ar@{-}[r]&3\ar@{-}[r]&1\ar@{-}[d]\\
				4\ar@{-}[u]&&&&4&3\ar@{-}[u]&&&&3\ar@{-}[d]\\
				3\ar@{-}[u]&&&&3\ar@{-}[d]\ar@{-}[u]&4\ar@{-}[u]&&&&4\ar@{-}[d]\\
				1\ar@{-}[u]&3\ar@{-}[l]&4\ar@{-}[l]&1\ar@{-}[l]&2\ar@{-}[l]&1\ar@{-}[u]&4\ar@{-}[l]&3\ar@{-}[l]&1\ar@{-}[l]&2\ar@{-}[l]
			}$
		\end{center}
		\caption{{Two circular sequences. These two sequences are different, but any three genes form a common subsequence. This counterexample fails Conjecture~\ref{c2}.}}
		\label{cec2}
	\end{figure}
	
	Although Conjecture~\ref{c2} has counterexamples, a weaker version does hold:
	
	\begin{lemma}
		In Scenario 4, if any three genes $g_i,g_j,g_l$ in $g_1,\ldots,g_k$ form a common subsequence, and one gene $g_s$ in $g_1,\ldots,g_k$ has only one copy in each sequence, then $g_1,\ldots,g_k$ form a common subsequence. 
		\label{lastlemma}
	\end{lemma}
	\begin{proof}
		The proof is similar to that of Lemma~\ref{lenew}. Only consider copies of $g_1,\ldots,g_l$ in these circular sequences. If $g_1,\ldots,g_l$ do not form a common subsequence, cut at $g_s$ to expand the circular sequences to linear sequence, and find the first digit that such linear sequences differ. Assume $g_p$ and $g_q$ can both appear in this digit. Then the copy sequences of $g_p,g_q$ are the same before this digit, and differ at this digit. This means $g_s,g_p,g_q$ cannot form a common subsequence for the linear sequences. Since $g_s$ has only one copy, and the linear subsequences of $g_s,g_p,g_q$ start with $g_s$, this means that $g_s,g_p,g_q$ cannot form a common subsequence for the circular sequences.
	\end{proof}
	
}

To solve Scenario 4 for $m$ sequences with $n$ genes, construct an auxiliary $3$-uniform hypergraph $\mathcal{G}$ as following \cite{Diestel}: vertices are genes $1,\ldots,n$; there is a $3$-hyperedge (undirected) that links genes $g_i,g_j,g_k$ if and only if they form a common subsequence. 

For a $3$-uniform hypergraph $\mathcal{G}$, a clique is a subset of vertices, such that each vertex triple in this subset is linked by a $3$-hyperedge. The maximum clique is the clique that has the most vertices.

{
	Since Conjecture~\ref{c2} does not hold, in Scenario 4, not every clique corresponds to a common subsequence. However, for each common subsequence with genes $g_1,\ldots,g_q$, since any gene triple $g_i,g_j,g_l$ in $g_1,\ldots,g_q$ form a subsequence, there is a $3$-hyperedge linking $g_i,g_j,g_l$. Therefore, the common subsequence $g_1,\ldots,g_q$ corresponds to a clique.

	\begin{proposition}
		For $m$ sequences with $n$ genes, if the maximum clique $g_1,\ldots,g_k$ of the auxiliary hypergraph $\mathcal{G}$ contains a gene $g_s$ that has only one copy in each sequence, then the maximum clique $g_1,\ldots,g_k$ corresponds to the LCS of the sequences. 
		\label{p3}
	\end{proposition}
	
	\begin{proof}
		From Lemma~\ref{lastlemma}, the maximum clique $g_1,\ldots,g_k$ corresponds to a common subsequence of length $k$. If the LCS has length $p>k$, then it corresponds to a clique of size $p$, larger than the maximum clique, a contradiction. Therefore, the LCS has length $k$, and the maximum clique $g_1,\ldots,g_k$ corresponds to the LCS.
	\end{proof}
	
	For many prokaryotes, more than half of the genes only have one copy in the gene sequence \cite{chern2011comparison}. Besides, if the maximum clique has at most $n/2$ genes, then those sequences are too different. In this case, determining transposable genes by finding the LCS does not make much sense in biology, since the actual evolution trajectory might not be the shortest one. Therefore, in practice, it is common for the maximum clique to contain a gene that has only one copy in each sequence, and Proposition~\ref{p3} generally holds.
}

Under a weak condition, we have reduced Scenario 4 into the maximum clique problem for $3$-uniform hypergraphs, which is also NP-hard \cite{wu2015review}. There have been some algorithms for the maximum clique problem for $3$-uniform hypergraphs \cite{torres2017hclique,rota2007continuous}. For completeness, we propose a simple idea: Repeatedly delete the gene that has the smallest degree, until we have a clique that any three genes have a $3$-hyperedge that links them. We summarize this greedy heuristic method as Algorithm~\ref{alg7}. This algorithm is easy to understand, and can provide some intuition. 

Algorithm~\ref{alg7} does not always produce the correct result. See Fig.~\ref{ce3} for a counterexample. Here each gene in $1,2,3,4,5,6$ has degree $4$, while each gene in $7,8,9,10$ has degree $3$ .When applying Algorithm~\ref{alg7}, genes $7,8,9,10$ are first deleted, and the final result just has three genes, such as $(1,3,5)$. However, the LCS $(7,8,9,10)$ has four genes.

\begin{figure}[htb]
	\begin{center}
		$\xymatrix{
			1\ar@{-}[r]&2\ar@{-}[r]&7\ar@{-}[r]&3\ar@{-}[r]&4\ar@{-}[d]&2\ar@{-}[r]&1\ar@{-}[r]&10\ar@{-}[r]&4\ar@{-}[r]&3\ar@{-}[d]\\
			10\ar@{-}[u]&6\ar@{-}[l]&5\ar@{-}[l]&9\ar@{-}[l]&8\ar@{-}[l]&9\ar@{-}[u]&5\ar@{-}[l]&6\ar@{-}[l]&8\ar@{-}[l]&7\ar@{-}[l]\\
			1\ar@{-}[r]&2\ar@{-}[r]&9\ar@{-}[r]&3\ar@{-}[r]&4\ar@{-}[d]&2\ar@{-}[r]&1\ar@{-}[r]&8\ar@{-}[r]&4\ar@{-}[r]&3\ar@{-}[d]\\
			8\ar@{-}[u]&6\ar@{-}[l]&5\ar@{-}[l]&7\ar@{-}[l]&10\ar@{-}[l]&7\ar@{-}[u]&5\ar@{-}[l]&6\ar@{-}[l]&10\ar@{-}[l]&9\ar@{-}[l]
		}$
	\end{center}
	\caption{Four circular sequences. The LCS is $(7,8,9,10)$. This counterexample fails Algorithm~\ref{alg7}.}
	\label{ce3}
\end{figure}
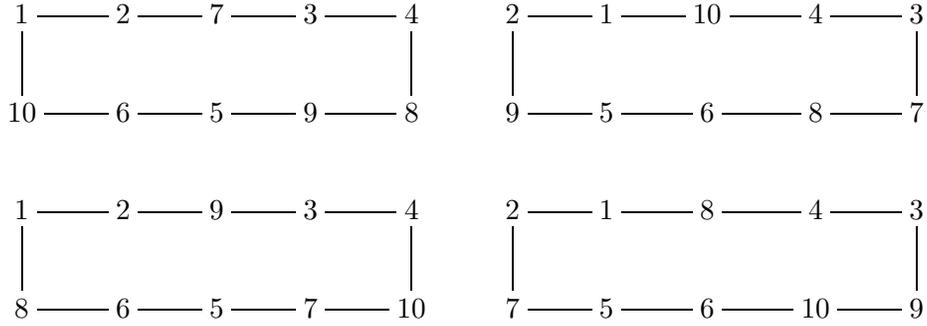

{We test Algorithm~\ref{alg7} on random graphs. Construct a random $3$-uniform hypergraph with $n$ genes, and any three genes have probability $p$ to have a $3$-hyperedge among them. Use brute-force search to find the maximum clique, and compare its size with the result of Algorithm~\ref{alg7}. For each $5\le n\le 12$ and $p=0.01,0.1,0.3,0.5,0.7,0.9,0.99$, we repeat this for $10000$ times. Since finding the true maximum clique requires exponentially slow brute-force search, we do not test on very large graphs. See Table~\ref{random4} for the success rate of Algorithm~\ref{alg7}. We can see that when $p\sim 0.5$, the success rate decreases rapidly with the hypergraph size $n$. In practice, different gene sequences do not differ too much, and most gene triples have a $3$-hyperedge that links them. This corresponds to $p\sim 1$ case, and Algorithm~\ref{alg7} can generally produce the correct result. Still, the performance of Algorithm~\ref{alg7} is not as good as other more sophisticated yet more complicated methods \cite{torres2017hclique,rota2007continuous}.}

\begin{table}[]
	\begin{tabular}{llllllll}
		\multirow{2}{*}{$n$} &    &    &    & \multicolumn{3}{l}{Success rate} &  \\
		& $p=0.01$ & $p=0.1$ & $p=0.3$ & $p=0.5$ & $p=0.7$ & $p=0.9$ & $p=0.99$ \\
		5                  & 1.0000   & 1.0000   & 1.0000   & 1.0000   & 1.0000   &  1.0000  & 1.0000   \\
		6                  & 1.0000   & 1.0000   & 0.9971   & 0.9676   & 0.9741   &  0.9998  & 1.0000   \\
		7                  & 1.0000   & 1.0000   & 0.9720   & 0.9167   & 0.9652   &  0.9854  & 1.0000   \\
		8                  & 1.0000   & 0.9996   & 0.9217   & 0.9089   & 0.9142   &  0.9559  & 0.9999  \\
		9                  & 1.0000   & 0.9993   & 0.8381   & 0.9210   & 0.8537   &  0.9255  & 0.9997   \\
		10                 &  1.0000  & 0.9969   & 0.7637   & 0.9190   & 0.8439   &  0.8980  &  0.9973  \\
		11                 &  1.0000  & 0.9947   & 0.7138   & 0.8737   &  0.8406  &  0.8402  &   0.9930 \\
		12                 &  1.0000  & 0.9893   & 0.7111   & 0.7890   & 0.7919   &  0.8177  &  0.9864 
	\end{tabular}
	\caption{{The test results of Algorithm~\ref{alg7} on random graphs. We generate a $3$-uniform hypergraph with $n$ vertices, and for each vertex triple, there is a $3$-hyperedge linking them with probability $p$. For each random graph, we check whether Algorithm~\ref{alg7} produces the true maximum clique. Each situation is repeated 10000 times, and we present the success rate.}}
	\label{random4}
\end{table}

{We apply Algorithm~\ref{alg7} to real gene sequencing data and manually confirm that the results are correct. See Appendix~\ref{app4}.}

Assume we have $m$ sequences with $n$ genes. In general, the copy number of a gene is small, and we can assume the length of each sequence is $\mathcal{O}(n)$. The time complexities of Step 2 and Step 3 in Algorithm~\ref{alg7} are $\mathcal{O}(mn^3)$ and $\mathcal{O}(n^3)$, and the overall time complexity is $\mathcal{O}(mn^3)$. The space complexity is trivially $\mathcal{O}(mn+n^3)$. 

{For the original LCS problem, since we only consider subsequences that contain all or none copies of the same gene, we do not know other algorithms that directly solve the same problem.
	
	Under a weak condition, we have transformed the LCS problem into the maximum clique problem for a $3$-uniform hypergraph. If the condition in Proposition~\ref{p3} does not hold, we do not know how Scenario 4 can be tackled. For the maximum clique problem, we recommend applying other algorithms for this maximum clique problem \cite{torres2017hclique,rota2007continuous}, not Algorithm~\ref{alg7}, due to its low success rate for larger graphs. Our main contribution for this scenario is Propositions~\ref{p2},\ref{p3}.}

\begin{algorithm}[!htbp]
	\caption{A heuristic method for detecting transposable genes in Scenario 4.}
	\label{alg7}
	\vspace{-\bigskipamount}
	\ \\
	\begin{enumerate}
		{	\item \textbf{Input}
			
			\quad $m$ circular sequences of genes $1,\ldots,n$, where each gene can have multiple copies
			
			\item \textbf{Construct} the auxiliary graph $\mathcal{G}$:
			
			\quad Vertices of $\mathcal{G}$ are all the genes $1,\ldots,n$ (not their copies)
			
			\quad \textbf{For} each gene triple $g_i,g_j,g_k$
			
			\quad\quad   \textbf{If} all copies of $g_i,g_j,g_k$ keep their relative locations in all $m$ sequences
			
			\quad\quad\quad \textbf{Add} a $3$-hyperedge that links $g_i,g_j,g_k$ in $\mathcal{G}$
			
			\quad\quad\textbf{End} of if
			
			\quad\textbf{End} of for			
			
			\item \textbf{While} there exist three genes that do not share a $3$-hyperedge
			
			\quad \textbf{Calculate} the degree for each gene in $\mathcal{G}$
			
			\quad\textbf{Delete} the gene with the smallest degree and $3$-hyperedges that links this gene
			
			\quad \% If there are multiple genes with the smallest degree, delete one randomly
			
			\textbf{End} of while
			
			\% After this while loop, any three genes form a common subsequence
			
			\% If the condition in Proposition~\ref{p3} holds, the remaining genes form a common subsequence
			
			\item \textbf{Output} remaining genes are not transposable, and other genes are transposable
		}
	\end{enumerate}
\end{algorithm}

\section{Conclusion and discussion}
\label{con}
In this paper, we study the LCS problem and design Algorithms~\ref{alg1}--\ref{alg7} for different scenarios. Specifically, we consider the case where the LCS is not unique, and determine whether each number appears in all/some/none of the LCSs. These algorithms are applied to gene sequences to determine the stability of genes. To apply those algorithms, one needs to apply genomic annotation tools to transform raw DNA sequencing data into gene sequences, and replace gene names by numbers. Those algorithms have at most $O(mn^3)$ time complexity, where $m$ is the number of sequences, and $n$ is the number of genes. Thus they can run in a reasonable time for most applications. We prove that the latter two scenarios are NP-hard (Corollary~\ref{coro} and Proposition~\ref{p2}).

We start with gene sequences and determine translocated genes. Therefore, short transposons (possibly shorter than a gene) cannot be determined. Besides, we do not determine specific genomic rearrangement events. We aim at determining which genes are able to translocate (i.e., less stable). Specifically, we study how many LCSs contain a certain gene, as a measure for its ``stability''. This mesoscopic viewpoint can be intriguing for understanding changes in genome.

We can adopt a stricter definition of transposable genes to exclude a gene which only changes its relative position in a few (no more than $l$, where $l$ is small enough) sequences. Then we should consider the longest sequence which is a common subsequence of at least $m-l$ sequences. We can run the corresponding algorithm for every $m-l$ sequences. Thus the total time complexity will be multiplied by a factor of $m^l$.

The results in this paper are not limited to Scenarios 1--4. They can be applied to other bioinformatics situations, or even other fields that need discrete mathematics tools, such as text processing, compiler optimization, data analysis, image analysis \cite{hajiaghayi2019massively}. Besides, algorithms in this paper might be able to detect non-syntenic regions \cite{lee2010non}.

There are some possible future directions: (1) extend Proposition~\ref{p3} to find more efficient solutions to Scenario 4; (2) determine whether genes appear in all LCSs in other similar scenarios; (3) develop fast algorithms when adopting a stricter definition of transposable genes.

\appendix
{
	\section{Applications on experimental data}
	\label{app}
	We apply Algorithms~\ref{alg1}--\ref{alg7} on \emph{Escherichia coli} gene sequences to determine transposons. From NCBI sequencing database, we obtain gene sequences of three individuals of \emph{E. coli} strain ST540 (GenBank CP007265.1, GenBank CP007390.1, GenBank CP007391.1) and three individuals of \emph{E. coli} strain ST2747 (GenBank CP007392.1, GenBank CP007393.1, GenBank CP007394.1). Notice that the gene sequence of \emph{E. coli} is circular. 
	
	\subsection{Scenario 1}
	All three sequences of ST540 start with gene dnaA and end with gene rpmH. We can regard them as linear gene sequences. We remove genes that appear more than once in one sequence, and remove genes that do not appear in all three sequences. After applying Algorithms~\ref{alg1},\ref{alg2} on these three sequences, there are 301 non-transposable genes, 4 quasi-transposable genes (hpaC, iraD, fbpC, psiB), and 263 proper-transposable genes. The reason for the large amount of proper-transposable genes is that sequence CP007265.1 is significantly different from the other two. After removing it and applying Algorithms~\ref{alg1},\ref{alg2} to the remaining two sequences (CP007390.1 and CP007391.1), there are 564 non-transposable genes and 4 quasi-transposable genes (hpaC, iraD, fbpC, psiB). Therefore, some genes in hpaC, iraD, fbpC, psiB are likely to translocate.
	
	All three sequences of ST2747 start with gene glnG and end with gene hemG. We can regard them as linear gene sequences. We remove genes that appear more than once in one sequence, and remove genes that do not appear in all three sequences. After applying Algorithms~\ref{alg1},\ref{alg2} on these three sequences, all 573 genes are non-transposable.
	
	\subsection{Scenario 2}
	
	We regard all three sequences of ST540 as circular gene sequences. We remove genes that appear more than once in one sequence, and remove genes that do not appear in all three sequences. After applying Algorithms~\ref{alg3},\ref{alg4} on these three sequences, there are 389 non-transposable genes, 50 quasi-transposable genes, and 129 proper-transposable genes. The reason for the large amount of proper-transposable genes is that sequence CP007265.1 is significantly different from the other two. After removing it and applying Algorithms~\ref{alg3},\ref{alg4} to the remaining two sequences (CP007390.1 and CP007391.1), there are 564 non-transposable genes and 4 quasi-transposable genes (hpaC, iraD, fbpC, psiB). Therefore, some genes in hpaC, iraD, fbpC, psiB are likely to translocate.
	
	We regard all three sequences of ST2747 as circular gene sequences. We remove genes that appear more than once in one sequence, and remove genes that do not appear in all three sequences. After applying Algorithms~\ref{alg3},\ref{alg4} on these three sequences, all 573 genes are non-transposable genes.
	
	\subsection{Scenario 3}
	\label{app3}
	All three sequences of ST540 start with gene dnaA and end with gene rpmH. We can regard them as linear gene sequences. We remove genes that do not appear in all three sequences. After applying Algorithm~\ref{alg5} on these three sequences, there are 308 non-transposable genes and 288 transposable genes. The reason for the large amount of transposable genes is that sequence CP007265.1 is significantly different from the other two. After removing it and applying Algorithm~\ref{alg5} to the remaining two sequences (CP007390.1 and CP007391.1), there are 582 non-transposable genes and 14 transposable genes (dkgB, metN, metQ, rcsF, prfC, rimI, rsmC, hpaC, iraD, fecC, fecD, purH, eutA, pssA). Notice that hpaC and iraD are already detected as quasi-transposable genes in Scenario 1. For the result of comparing CP007390.1 and CP007391.1, we manually confirm that this result is correct, since those 582 non-transposable genes form a common subsequence, 11 transposable genes have different copy numbers in different sequences, and the other 3 transposable genes (hpaC, iraD, pssA) cannot be used to form a longer common subsequence.
	
	All three sequences of ST2747 start with gene glnG and end with gene hemG. We can regard them as linear gene sequences. We remove genes that do not appear in all three sequences. After applying Algorithm~\ref{alg5} on these three sequences, there are 582 non-transposable genes and 15 transposable genes (mobA, btuB, lacY, sdhB, ihfB, rnfD, tppB, emrB, rluD, fmt, trkA, mscL, nepI, hemD, hemC). We manually confirm that this result is correct, since those 582 non-transposable genes form a common subsequence, and the other 15 transposable genes have different copy numbers in different sequences.
	
	\subsection{Scenario 4}
	\label{app4}
	We regard all three sequences of ST540 as circular gene sequences. We remove genes that do not appear in all three sequences. After applying Algorithm~\ref{alg7} on these three sequences, there are 400 non-transposable genes and 196 transposable genes. The reason for the large amount of proper-transposable genes is that sequence CP007265.1 is significantly different from the other two. After removing it and applying Algorithm~\ref{alg7} to the remaining two sequences (CP007390.1 and CP007391.1), there are 582 non-transposable genes and 14 transposable genes (dkgB, metN, metQ, rcsF, prfC, rimI, rsmC, hpaC, iraD, fecC, fecD, purH, eutA, pssA). Notice that hpaC and iraD are already detected as quasi-transposable genes in Scenario 2. For the result of comparing CP007390.1 and CP007391.1, we manually confirm that this result is correct, since those 582 non-transposable genes form a common subsequence, 11 transposable genes have different copy numbers in different sequences, and the other 3 transposable genes (hpaC, iraD, pssA) cannot be used to form a longer common subsequence.
	
	We regard all three sequences of ST2747 as circular gene sequences. We remove genes that do not appear in all three sequences. After applying Algorithm~\ref{alg7} on these three sequences, there are 582 non-transposable genes and 15 transposable genes (mobA, btuB, lacY, sdhB, ihfB, rnfD, tppB, emrB, rluD, fmt, trkA, mscL, nepI, hemD, hemC). We manually confirm that this result is correct, since those 582 non-transposable genes form a common subsequence, and the other 15 transposable genes have different copy numbers in different sequences.}

\section*{Acknowledgments}
The author would like to thank Zhongkai Zhao for helping with designing Algorithm~\ref{alg1}. The author would like to thank Ao Sun for constructing a counterexample for Conjecture~\ref{c2}. The author would like to thank Lucas B\"ottcher, Jun Su, Boyu Zhang, Ruixiang Zhang, and anonymous reviewers for providing helpful comments.

\bibliographystyle{acm}
\bibliography{Transposons}

\end{document}